%% file: THzTrans_main.tex
\documentclass[conference, a4paper]{IEEEtran}

\usepackage{cite}
\usepackage[top = 0.8in, left = 0.66in, right = 0.66in, bottom = 1.253in]{geometry}
\usepackage{amsthm, amsmath,amssymb,amsfonts,bbm}
\usepackage{algorithmic}
\usepackage[dvips]{graphicx}
\usepackage{subfigure}

\graphicspath{{graphics/}, {graphics/Results}, {graphics/Results/0_MatlabResults}}
\usepackage{makecell}
\usepackage{multirow}
\usepackage{filecontents, pgfplots}
\usepackage[algoruled]{algorithm2e}
\usepackage{microtype}
\usepackage{tikz}
\usepackage{textcomp}
\usepackage{xcolor}
\usepackage{verbatim}
\usepackage[binary-units=true]{siunitx}
\DeclareSIUnit{\belmilliwatt}{Bm}
\DeclareSIUnit{\dBm}{\deci\belmilliwatt}
\DeclareSIUnit{\dBi}{dBi}

\DeclareMathOperator*{\maximize}{maximize}

\DeclareMathOperator*{\subjectto}{subject\,to}

\usepackage[acronym,nomain]{glossaries}
\input{Glossary}

\begin{document}

	\newtheorem{proposition}{Proposition}	
	\newtheorem{lemma}{Lemma}	
	\newtheorem{corollary}{Corollary}
	\newtheorem{assumption}{Assumption}	
	\newtheorem{remark}{Remark}	
	
	\title{Resonant Tunneling Diode-Based THz SWIPT for Microscopic 6G IoT Devices}

	\author{\IEEEauthorblockN{Nikita Shanin\IEEEauthorrefmark{1}, Simone Clochiatti\IEEEauthorrefmark{2}, Kenneth M. Mayer\IEEEauthorrefmark{1}, Laura Cottatellucci\IEEEauthorrefmark{1}, \\Nils Weimann\IEEEauthorrefmark{2}, and Robert Schober\IEEEauthorrefmark{1}}\\\vspace*{-5pt}
		\IEEEauthorblockA{\IEEEauthorrefmark{1}\textit{Friedrich-Alexander-Universit\"{a}t Erlangen-N\"{u}rnberg (FAU), Germany}\\ 
			\IEEEauthorrefmark{2}\textit{Universit\"{a}t Duisburg-Essen, Germany} } }

	\maketitle

\begin{abstract}
	\input{Abstract}
\end{abstract}
\vspace*{-5pt}
\section{Introduction}
\label{Section:Introduction}
\input{IntroductionSection}

\section{System Model}
\label{Section:SysModel}
\input{SystemModelSection}

\paragraph*{Electrical circuit at the RX}
\input{ElectricalCircuit}

\paragraph*{EH model for the RTD-based RX circuit}
\input{RTD_EH}

\vspace*{-5pt}
\section{Problem Formulation and Solution}
\label{Section:Problem}
\input{Problem}

\vspace*{-7pt}
\section{Numerical Results}
\label{Section:Results}
\input{SystemSetup}

\section{Conclusions}
\input{Conclusions}

\appendices
	\renewcommand{\thesection}{\Alph{section}}
	\renewcommand{\thesubsection}{\thesection.\arabic{subsection}}
	\renewcommand{\thesectiondis}[2]{\Alph{section}:}
	\renewcommand{\thesubsectiondis}{\thesection.\arabic{subsection}:}	
\vspace*{-5pt}
	\section{Proof of Proposition \ref{Prop:Feasibility}}
	\label{Appendix:PropFeasibility}
	\input{ProofFeasibility}
	\vspace*{-5pt}
	\section{Proof of Proposition \ref{Prop:UniformDistribution}}
	\label{Appendix:UniformDistribution}
	\input{ProofUniformDistribution}
	\vspace*{-5pt}
	\section{Proof of Proposition \ref{Prop:OptimalSolution}}
	\label{Appendix:OptimalSolution}
	\input{ProofOptimalSolution}

\vspace*{-7pt}
\bibliographystyle{IEEEtran}
\bibliography{WPT_THz.bib}

\end{document}

%% file: Glossary.tex
\makeglossaries
\newacronym{swipt}{SWIPT}{simultaneous wireless information and power transfer}
\newacronym{wpt}{WPT}{wireless power transfer}
\newacronym{wpcn}{WPCN}{wireless powered communication network}
\glsunset{wpcn}
\newacronym{wit}{WIT}{wireless information transfer}
\newacronym{awgn}{AWGN}{additive white Gaussian noise}
\newacronym{tx}{TX}{transmitter}
\newacronym{rx}{RX}{receiver}
\newacronym{bs}{BS}{base station}
\newacronym{ir}{IR}{information receiver}
\newacronym{eh}{EH}{energy harvesting}
\newacronym{id}{ID}{information detection}
\newacronym{irs}{IRS}{intelligent reflected surface}
\newacronym{ap}{AP}{average power}
\newacronym{pp}{PP}{peak power}
\newacronym{aoa}{AOA}{angle-of-arrival}
\newacronym{aod}{AOD}{angle-of-departure}
\newacronym{eec}{EEC}{electical equivalent circuit}

\newacronym{fso}{FSO}{free space optical}
\glsunset{fso}
\newacronym{vlc}{VLC}{visible light communication}
\newacronym{led}{LED}{light emitting diode}
\glsunset{led}
\newacronym{ook}{OOK}{On-Off keying}
\newacronym{pam}{PAM}{pulse amplitude modulated}
\newacronym{tdma}{TDMA}{time division multiple access}
\newacronym{ml}{ML}{maximum likelihood}

\newacronym{siso}{SISO}{single-input single-output}
\newacronym{mimo}{MIMO}{multiple-input multiple-output}
\newacronym{miso}{MISO}{multiple-input single-output}
\newacronym{simo}{SIMO}{single-input multiple-output}

\newacronym{rf}{RF}{radio frequency}
\newacronym{dc}{DC}{direct current}
\newacronym{ac}{AC}{alternating current}
\newacronym{papr}{PAPR}{peak-to-average power ratio}
\newacronym{lp}{LPF}{low-pass filter}
\newacronym{mc}{MC}{matching circuit}
\newacronym{rtd}{RTD}{resonant-tunelling diode}
\newacronym{mrt}{MRT}{maximum ratio transmission}
\newacronym{ecb}{ECB}{equivalent complex baseband}
\newacronym{tdd}{TDD}{time-division-duplex}
\newacronym{zf}{ZF}{zero forcing}
\newacronym{snr}{SNR}{signal-to-noise ratio}
\newacronym{sinr}{SINR}{signal-to-interference-plus-noise ratio}

\newacronym{rv}{RV}{random variable}
\newacronym{iid}{i.i.d.}{independent and identically distributed}
\newacronym{Pdf}{Pdf}{Probability density function}
\newacronym{pdf}{pdf}{probability density function}
\newacronym{cdf}{cdf}{cumulative density function}

\newacronym{dnn}{DNN}{dense neural networks}
\glsunset{dnn}
\newacronym{mdp}{MDP}{Markov decision process}

\newacronym{sca}{SCA}{successive convex approximation}
\newacronym{sdr}{SDR}{semi-definite relaxation}

\newacronym{spr}{LP}{low power}
\newacronym{mpr}{MP}{medium power}
\newacronym{lpr}{HP}{high power}

%% file: Abstract.tex
In this paper, we study terahertz (THz) simultaneous wireless information and power transfer (SWIPT) for future micro-scale 6G Internet-of-Things (IoT) networks.
Since Schottky diodes are not efficient for THz energy harvesting (EH), we propose resonant tunneling diodes (RTDs) for EH at the IoT receiver (RX). 
As the electrical properties of RTDs are different from those of Schottky diodes, we develop a novel closed-form EH model for RTD-based RXs.
In particular, we model the dependency of the instantaneous RX output power on the instantaneous received power by a non-linear piecewise function, whose parameters are adjusted to fit circuit simulation results.
Furthermore, since coherent information detection is challenging at THz frequencies, we employ unipolar amplitude shift keying (ASK) modulation at the transmitter (TX) and utilize the RTD-based EH circuit at the RX to extract both information and energy from the received signal.
We formulate an optimization problem to maximize the mutual information between the TX and RX signals subject to constraints on the peak amplitude of the transmitted signal and the required average harvested power at the RX.
Moreover, we determine a feasibility condition for the formulated problem and, for high and low required average harvested powers, we derive the achievable information rate numerically and in closed form, respectively.
Our simulation results highlight a tradeoff between the information rate and the average harvested power.
Finally, we show that this tradeoff is determined by the peak amplitude of the transmitted signal and the maximum instantaneous harvested power for low and high received signal powers, respectively.

%% file: IntroductionSection.tex
One of the most disruptive use cases of future $6\text{G}$ communication systems are micro-scale networks, e.g., for biomedical, wearable, and industrial applications, where ultra-small low-power Internet-of-Things (IoT) devices communicate at extremely high data rates exceeding $\SI{100}{\giga\bit\per\second}$ \cite{Tataria2021, Kuscu2021, Yi2021}.
A key technology for these micro-scale networks is terahertz (THz)-band communication, where the required extremely high data rates are enabled by the huge available spectrum \cite{Tataria2021, Yi2021}.
Although THz-band antennas can be made ultra-small, the need to regularly replace bulky batteries does not allow the design of truly microscopic THz IoT devices.
A promising solution to this non-trivial problem is \gls*{swipt}, where both information and power are transmitted to the IoT devices in the downlink making battery replacements unnecessary \cite{Zhang2013, Kim2022, Clerckx2019, Boshkovska2015, Morsi2019}.

\gls*{swipt} systems utilizing \gls*{rf} GHz-band signals were considered in \cite{Zhang2013}.
Since \gls*{eh} circuits are not suitable for decoding information when phase modulation is utilized for information transmission, the authors of \cite{Zhang2013} proposed time-sharing and power-splitting between \gls*{eh} and information detection at the \gls*{rx} and showed a tradeoff between the achievable data rate and the average power harvested at the user device.
For the system design in \cite{Zhang2013}, the authors assumed a linear relationship between the received and harvested powers at the EH nodes.
However, the experimental results in \cite{Kim2022} validated that practical electrical circuits utilized for \gls*{eh} exhibit a highly non-linear behavior.
In fact, EH circuits are typically composed of a matching network, a non-linear rectifying diode, a low-pass filter, and a load resistance \cite{Clerckx2019}.
The non-linearity of EH circuits is determined by the non-linear forward-bias current-voltage (I-V) characteristic of the rectifying diode for low input powers, while, in the high input power regime, EH circuits are driven into saturation due to the breakdown of the employed diode \cite{Clerckx2019}.
To take the non-linear behavior of EH circuits into account, in \cite{Boshkovska2015} and \cite{Morsi2019}, the authors proposed non-linear EH models for the design of EH communication networks.
The authors of \cite{Boshkovska2015} proposed an EH model based on a parameterized sigmoidal function to describe the dependency between the \textit{average} harvested power and the \textit{average} input power at the EH node.
The parameters of the EH model in \cite{Boshkovska2015} were adjusted to fit circuit simulation results assuming a Gaussian distribution for the transmit signal.
Since the model in \cite{Boshkovska2015} does not allow the optimization of the transmit signal distribution, the authors of \cite{Morsi2019} accurately analyzed a single Schottky diode-based EH circuit, as they are typically utilized for RF EH, and derived a corresponding closed-form EH model to characterize the \textit{instantaneous} harvested power at a user device as a function of the \textit{instantaneous} received signal power.
Utilizing the proposed circuit-based EH model, in \cite{Morsi2019}, the authors studied SWIPT systems operating in the GHz frequency band with separate energy and information RXs and determined the optimal transmit signal distribution that maximizes the mutual information between the transmit signal and the signal received at the information RX under a constraint on the average harvested power at the EH node.

Although the EH model in \cite{Morsi2019} accurately characterizes the instantaneous harvested power at EH circuits employing a Schottky diode, it is not suitable for the signal design for SWIPT systems operating in the THz frequency band.
In fact, not only the characteristics of the wireless communication channels, but also the properties of the utilized electronic devices vary with frequency.
In particular, the Schottky diodes employed in EH circuits in the GHz frequency band \cite{Kim2022, Clerckx2019, Boshkovska2015, Morsi2019} are not efficient for EH and signal demodulation in THz SWIPT systems \cite{Villani2021}.
On the other hand, resonant tunneling diodes (RTDs) featuring a potential well and multiple quantum barriers have much smaller transient times and a larger curvature at the zero-bias point compared to Schottky diodes, which allows them to efficiently operate up to frequencies of $\SI{2}{\tera\hertz}$ \cite{Clochiatti2020, Villani2021}.
However, the experimental results in \cite{Villani2021} and \cite{Clochiatti2020} showed that the I-V forward-bias characteristic of RTDs is not only highly non-linear, but also exhibits multiple critical points, regions of negative resistance, and a dependency on the signal frequency.
Thus, RTDs differ substantially in their behavior from Schottky diodes which exhibit static and monotonic I-V characteristics \cite{Tietze2012}.
Although a complete and accurate model capturing the frequency dependency, non-linearity, and non-monotonicity of the current flow through an RTD is not available yet, in \cite{Clochiatti2020}, the authors developed a Keysight ADS \cite{ADS2017} design of an RTD, which fits the measurement data presented in \cite{Clochiatti2020} for a wide range of operating frequencies.
Finally, we note that unlike for GHz-band SWIPT, the design of a coherent information IoT \gls*{rx}, which detects the phase of the received THz signal, is challenging due to the instability and phase noise of THz local oscillators \cite{Yi2021}.
However, since the spectrum available in the THz band is significantly larger than in the GHz frequency band, communication systems employing unipolar amplitude shift keying (ASK) modulation are also able to achieve high data rates \cite{Yi2021}. 
Furthermore, since EH circuits are envelope detectors, EH circuits do not only harvest energy, but can also be exploited to extract the transmitted information from the received THz signal if unipolar ASK modulation is utilized.
To the best of the authors' knowledge, this is the first work to study THz SWIPT systems enabled by RTD-based EH.

In this paper, we study single-user THz SWIPT systems for 6G IoT networks.
The main contributions of this work can be summarized as follows.
We propose RTDs for EH at the IoT RX, and to characterize the instantaneous power of the output signal at the RX as a function of the instantaneous received power, we develop a novel non-linear piecewise EH model, whose parameters are adjusted to fit Keysight ADS circuit simulation results \cite{Clochiatti2020, ADS2017}.
Furthermore, we employ unipolar ASK modulation at the \gls*{tx} and utilize the RTD-based EH circuit at the RX to extract both information and energy from the received signal.
Based on the proposed EH model, we formulate an optimization problem for the maximization of the mutual information between the TX and RX signals subject to constraints on the peak amplitude of the transmitted signal and the average harvested power at the RX.
Moreover, we provide a feasibility condition for this optimization problem, and for high and low required average harvested powers, we determine the achievable information rate numerically and in closed form, respectively.
Our simulation results demonstrate a tradeoff between the information rate and the average harvested power.
Finally, we show that this tradeoff is determined by the peak amplitude of the transmitted signal and the maximum instantaneous harvested power for low and high received signal powers, respectively.

Throughout this paper, we use the following {\textit{notations}}.
We denote the sets of real, real non-negative, and non-negative integer numbers as $\mathbb{R}$, $\mathbb{R}_{+}$, and $\mathbb{N}$, respectively.
The real-part of a complex variable $x$ is denoted by $\mathcal{R} \{x\}$, whereas $j = \sqrt{-1}$ is the imaginary unit.
Function $f_s(s)$ denotes the \gls*{pdf} of random variable $s$.
$\mathbb{E}_s \{\cdot\}$ stands for statistical expectation with respect to random variable $s$.
The domain and first-order derivative of one-dimensional function $f(\cdot)$ are denoted by $\mathcal{D}\{f\}$ and $f'(\cdot)$, respectively.

%% file: SystemModelSection.tex
\begin{figure}[!t]
	\centering
	\includegraphics[draft = false, width = 0.4\textwidth]{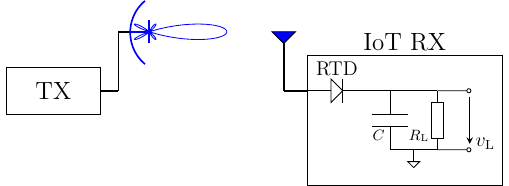}
	\vspace*{-7pt}
	\caption{THz SWIPT system model employing a TX with a highly directive antenna and a single-antenna IoT RX equipped with an \gls*{rtd}-based EH circuit.}
	\label{Fig:SystemModel}
	\vspace*{-10pt}
\end{figure}

We consider the \gls*{swipt} system in Fig.~\ref{Fig:SystemModel} where a \gls*{tx} equipped with a highly directive antenna sends a THz signal to a single-antenna microscopic IoT \gls*{rx}.
Since, in contrast to GHz-band signals, coherent demodulation of a THz-band received signal is challenging due to the instability and phase noise of the THz local oscillators, the TX employs unipolar ASK modulation for SWIPT \cite{Yi2021, Lemic2020}.
Thus, the THz signal $r(t) \in \mathbb{R}$ received at the RX can be expressed as follows:\vspace*{-3pt}
\begin{equation}
	r(t) = \sqrt{2} \mathcal{R} \{[h x(t) + n(t)] \exp(j2\pi f_c t) \},\vspace*{-3pt}
\end{equation}
\noindent where $h$ is the channel between TX and RX, which is assumed to be perfectly known at both devices, $f_c$ is the carrier frequency, and $x(t) = \sum_k s[k] \phi(t - kT)$ and $n(t)$ are \gls*{ecb} representations of the transmit signal and the \gls*{awgn} at the RX, respectively.
Furthermore, here, $s[k] \in \mathbb{R}_{+}, k \in \mathbb{N},$ are \gls*{iid} realizations of a non-negative random variable $s$ with \gls*{pdf} $f_s(s)$, $\phi(t)$ is a rectangular pulse that takes value $1$, if $t \in [0, T)$, and $0$, otherwise, and $T$ is the duration of a symbol interval.
To avoid signal distortion due to power amplifier non-linearities at the TX \cite{Morsi2019}, the maximum transmit signal amplitude is bounded by $A$.
Thus, the support of $f_s$ is confined to the interval $[0,A]$, i.e., $\mathcal{D}\{f_s\} \subseteq [0, A]$.

%% file: ElectricalCircuit.tex
Since Schottky diodes may not be efficient for EH and signal demodulation at THz frequencies \cite{Villani2021}, we adopt an electrical EH circuit comprising an RTD, a capacitance $C$, and a load resistance $R_\text{L}$ \cite{Clerckx2019, Morsi2019}, as shown in Fig.~\ref{Fig:SystemModel}.
Neglecting the ripples of the output voltage caused by the non-ideality of the employed low-pass filter, we express the signal used for information decoding at the output of the RX $y[k] \in \mathbb{R}, \forall k \in \mathbb{N},$ as follows:
\begin{equation}
	y[k] \triangleq \frac{v_\text{L}[k]}{\sqrt{R_\text{L}}} = \sqrt{\psi( |h s[k]|^2) } + n[k], \vspace*{-5pt}
\end{equation}
\noindent where $v_\text{L}[k]$ is the \gls*{dc} voltage at the output of the EH circuit, i.e., at the resistance $R_\text{L}$, in time slot $k, k\in\mathbb{N},$ and $n[k]$ and $\psi(\cdot)$ are the equivalent output noise sample and the function that maps the received signal power $\rho[k] = | h s[k] |^2$ to the power harvested at the load $R_\text{L}$ in time slot $k, \forall k\in\mathbb{N},$ respectively.
We note that since unipolar ASK signals are adopted at the TX, the output RX signal is characterized by the input signal envelope and can be used\footnotemark\hspace*{0pt} not only for charging a load device, but also for information detection.
\footnotetext{In a practical RX design, one may use a DC voltage divider \cite{Tietze2012} to split the output DC current and utilize the portions $\nu y[k]$ and $(1-\nu) y[k]$ of the current flow $y[k], \forall k,$ for EH and information detection, respectively, where $\nu \in [0,1]$ is the current splitting ratio. However, since the performance of information detection does not depend on $\nu$, we assume that the DC current $y[k]$ is directly utilized for detection of the transmitted message.}
Function $\psi(\cdot)$, which models the RTD-based RX circuit in Fig.~\ref{Fig:SystemModel}, will be characterized in the next section.
The output noise $n[k]$ in time slot $k, k\in\mathbb{N},$ is composed of the received noise from external sources and the internal thermal noise generated by the components of the RX electrical circuit.
To characterize the performance of information decoding, we assume that the thermal noise originating from the RX components dominates and the dependency of $n[k]$ on the received power is negligible \cite{Lapidoth2009}.
Hence, we model the output noise samples $n[k], \forall k,$ as \gls*{iid} realizations of AWGN with zero mean and variance $\sigma^2$.

\begin{figure}[!t]
	\centering
	\includegraphics[draft = false, width=0.4\textwidth]{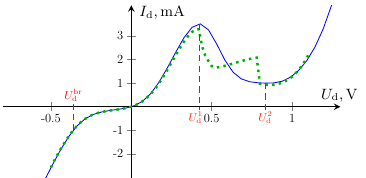}
	\vspace*{-10pt}
	\caption{I-V characteristic of the triple-barrier RTD design in \cite{Clochiatti2020} (blue solid line) matched to the measurement data from \cite{Clochiatti2020} (green dotted line).}
	\label{Fig:IV_Curve}
	\vspace*{-15pt}
\end{figure}
We note that an RTD has multiple quantum barriers, and thus, the I-V characteristic of an RTD includes regions of negative resistance \cite{Villani2021, Clochiatti2020}.
Specifically, in contrast to Schottky diodes, the current flow $I_\text{d}$ through the diode is not a monotonic increasing function of the applied voltage $U_\text{d}$, but the corresponding I-V curve may have multiple critical points, as, e.g., points\footnotemark\hspace*{0pt} $U^1_\text{d}$ and $U^2_\text{d}$ in Fig.~\ref{Fig:IV_Curve}.
Furthermore, for high input powers, the reverse-bias applied voltage $U_\text{d}$ may fall below the breakdown voltage $U_\text{d}^\text{br}$ and damage the device \cite{Clochiatti2020}.
\footnotetext{When the RTD is driven into a region of negative resistance, i.e., when $U_\text{d} \in [U^1_\text{d}, U^2_\text{d}]$, the RTD exhibits an additional RF gain \cite{Clochiatti2020}. The analysis of the impact of this RF gain on the efficiency of EH and information decoding is an interesting direction for further research and is beyond the scope of this paper.}
Since the current flow of an RTD depends on the signal frequency and is determined by quantum processes, an accurate analytical characterization of the I-V characteristic of an RTD is not available \cite{Villani2021}. 
Therefore, for our numerical results in Section IV, we employ the triple-barrier RTD design that was developed in \cite{Clochiatti2020} and fits the measured\footnotemark\hspace*{0pt} I-V characteristic and S-parameters of the diode designed in \cite{Clochiatti2020}, see Fig.~\ref{Fig:IV_Curve}. 
\footnotetext{We note that the discrepancy between the RTD I-V characteristic and measurement data for $U_\text{d} \in [U_\text{d}^1, U_\text{d}^2]$ in Fig.~\ref{Fig:IV_Curve} is caused by non-idealities (parasitic oscillations) of the measurement setup adopted in \cite{Clochiatti2020}.}

%% file: RTD_EH.tex
\begin{figure}[!t]
	\centering
	\includegraphics[draft = false, width=0.36\textwidth]{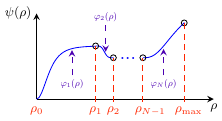}
	\vspace*{-10pt}
	\caption{Proposed general parameterized EH model for an RTD-based IoT RX.}
	\label{Fig:EH_Model}
	\vspace*{-5pt}
\end{figure}
In contrast to the Schottky diodes utilized for RF EH in \cite{Kim2022, Clerckx2019, Boshkovska2015, Morsi2019} and the related literature, the I-V characteristic of an RTD is not monotonic, as shown in Fig.~\ref{Fig:IV_Curve}.
Therefore, the instantaneous harvested power $P_\text{h} = \frac{v_\text{L}^2}{R_\text{L}}$ of an RTD-based EH circuit is also not a monotonic non-decreasing function of the input power $\rho = |hs|^2$ \cite{Villani2021, Clochiatti2020, Morsi2019}.
Moreover, since a closed-form expression for the current flow $I_\text{d}$ through an RTD is not available in the literature, the derivation of an accurate closed-form expression for the harvested power is not feasible \cite{Villani2021, Clochiatti2020}.
Therefore, as in \cite{Boshkovska2015}, in the following, we propose a \textit{general} \textit{parameterized} \textit{non-linear} EH model to characterize the instantaneous power harvested at the RX.
To this end, we model the dependency between the harvested power $P_\text{h}$ and the received signal power $\rho$ with the following piecewise function $\psi(\rho)$, see also Fig.~\ref{Fig:EH_Model}:\vspace*{-5pt}
\begin{equation}
	P_\text{h} \triangleq \psi (\rho) = \begin{cases}
		\varphi_1(\rho), \; &\text{if} \; \rho \in [\rho_0, \rho_1),\\
		\varphi_2(\rho), \; &\text{if} \; \rho \in [\rho_1, \rho_2),\\
		\cdots \\
		\varphi_N(\rho), \; &\text{if} \; \rho \in [\rho_{N-1}, \rho_\text{max}],
	\end{cases}
\label{Eqn:EHmodel}\vspace*{-3pt}
\end{equation} 
\noindent where function $\psi(\cdot)$ is defined in the domain $\mathcal{D}\{\psi\} = [\rho_0, \rho_\text{max}]$ and $\rho_\text{max}$ is the maximum value of the received signal power that does not drive the RTD into breakdown.
Here, the number $N \in \mathbb{N}$ of monotonic functions $\varphi_n(\cdot)$ with $\mathcal{D}\{\varphi_n\} = [\rho_{n-1}, \rho_n), n \in \{1,2,\cdots, N-1\},$ $\mathcal{D}\{\varphi_N\} = [\rho_{N-1}, \rho_\text{max}]$, and $0 \triangleq \rho_0 \leq \rho_1 \leq \rho_2 \leq \cdots \leq \rho_N \triangleq \rho_\text{max},$ that are needed for modelling $\psi(\cdot)$, depends on the number of critical points in the I-V characteristic and the breakdown voltage of the RTD.
Furthermore, since the I-V characteristic of an RTD alternates regions where $I_\text{d}$ increases and decreases when $U_\text{d}$ grows, as, for example, the intervals $(-\infty, U_\text{d}^1]$, $(U_\text{d}^1, U_\text{d}^2]$, and $(U_\text{d}^2, + \infty)$ in Fig.~\ref{Fig:IV_Curve}, we adopt parameterized monotonic increasing and decreasing functions $\varphi_n(\cdot)$ for odd and even values of $n$, i.e., $n \in \{1,3,\cdots\}$ and $n \in \{2,4,\cdots\}$ with $n \leq N$, respectively, as shown in Fig.~\ref{Fig:EH_Model}.
Finally, we express the average harvested power at the RX as a function of input pdf $f_s(s)$ as follows:\vspace*{-3pt}
\begin{equation}
	\bar{P}_\text{harv}(f_s) = \mathbb{E}_s \{ \psi(|hs|^2) \}. 
	\label{Eqn:AverageHarvestedPower} \vspace*{-3pt}
\end{equation}
\noindent Here, we neglect the impact of noise since its contribution to the average harvested power is negligible \cite{Boshkovska2015, Morsi2019}.

\paragraph*{Parameterized model of $\varphi_n(\cdot)$}
Since an accurate analytical derivation of functions $\varphi_n(\cdot), n\in\{1,2,\cdots, N\},$ is not feasible for an RTD-based EH circuit, for our numerical results in Section~\ref{Section:Results}, we model $\varphi_n(\cdot), \forall n,$ as a $5$-parameter logistic function as follows \cite{Gottschalk2005}:\vspace*{-5pt}
\begin{equation}
	\varphi_n(\rho) = B_n + (\Phi_n - B_n) \Big[ 1 + \theta_n (\rho - \rho_{n-1})^{\alpha_n} \Big]^{-\beta_n}.
	\label{Eqn:SigmoidalFunctionModel}\vspace*{-5pt}
\end{equation}
\noindent Here, $\Phi_n = \varphi_{n-1}(\rho_{n-1}), \forall n,$ with $\Phi_1 = 0$, $B_n = \lim_{\rho \to \infty} \varphi_n(\rho)$, and parameters $\alpha_n, \beta_n$, $\theta_n  \in \mathbb{R}_+$ characterize the non-linearity of $\varphi_n(\cdot), \forall n.$
As in \cite{Boshkovska2015}, these parameters can be obtained through a curve-fitting approach to optimize the matching between $\psi(\cdot)$ and measurement or simulation data.
Furthermore, as in \cite{Zhang2013, Kim2022, Clerckx2019, Boshkovska2015, Morsi2019}, we assume that all parameters of the RX circuit are perfectly known at the TX.

%% file: Problem.tex
In this section, we determine the tradeoff between the achievable information rate and the average harvested power at the RX.
To this end, we formulate the following optimization problem:\vspace*{-8pt}
\begin{subequations}
	\begin{align}
		\maximize_{f_s \in \mathcal{F}_{ \bar{A} }  }\quad \; &I(f_s) \label{Eqn:GeneralObj}\\
		\subjectto \quad\; & \bar{P}_\text{harv}(f_s) \geq \bar{P}^\text{req}_\text{harv}, \label{Eqn:GeneralOptConstr1}
	\end{align}
	\label{Eqn:GeneralOptimizationProblem}
\end{subequations}
\noindent\hspace*{-4pt}where we obtain the input pdf $f_s(\cdot)$ that maximizes the mutual information $I(f_s)$ between signals $s$ and $y$ subject to a constraint on the required average power\footnotemark\hspace*{0pt} $\bar{P}^\text{req}_\text{harv}$ harvested at the RX.
\footnotetext{The harvested power at the RX can be utilized, e.g., for sensing or signal processing tasks \cite{Morsi2019, Clerckx2019}. We assume that the harvested power $\bar{P}^\text{req}_\text{harv}$ needed to accomplish these tasks is known and a solution of optimization problem (\ref{Eqn:GeneralOptimizationProblem}) can be obtained at the TX.}
Here, $\mathcal{F}_{ \bar{A} } = \{f_s \; \vert \; \mathcal{D}\{f_s\} \subseteq [0, \bar{A}], \int_{s}f_s(s)\, \text{d}s = 1 \}$ denotes the set of feasible pdfs whose support does not exceed the maximum value $\bar{A} = \min \{A, \frac{ \sqrt{\rho_\text{max}} }{|h|}\}$, such that the transmit signal amplitude is upper-bounded by $A$ and the RTD is not driven into breakdown.

In the following proposition, we first determine the condition when (\ref{Eqn:GeneralOptimizationProblem}) is a feasible optimization problem.
\begin{proposition}
	For a given $\bar{A} = \min \{A, \frac{ \sqrt{\rho_\text{\upshape max}} }{|h|}\}$, optimization problem (\ref{Eqn:GeneralOptimizationProblem}) is feasible if and only if $\bar{P}^\text{\upshape req}_\text{\upshape harv} \in \Big[0, \bar{P}_\text{\upshape max} \Big]$ with $\bar{P}_\text{\upshape max} = \max_{\rho \in [0, |h\bar{A}|^2 ] } \psi(\rho)$.
	\label{Prop:Feasibility}
\end{proposition}
\begin{proof}
Please refer to Appendix~\ref{Appendix:PropFeasibility}.
\end{proof}

Proposition~\ref{Prop:Feasibility} highlights that for a given peak amplitude ${ \bar{A} }$, the average power harvested at the RX is bounded by $\bar{P}_\text{\upshape max}$, and hence, a solution of (\ref{Eqn:GeneralOptimizationProblem}) does not exist for $\bar{P}^\text{\upshape req}_\text{\upshape harv} > \bar{P}_\text{\upshape max}$.
Since determining the optimal pdf that solves (\ref{Eqn:GeneralOptimizationProblem}) is challenging, in the following, for a given $\bar{A}$ and $\bar{P}^\text{\upshape req}_\text{\upshape harv} \in [0, \bar{P}_\text{\upshape max}]$, we derive the achievable mutual information as a suboptimal solution of (\ref{Eqn:GeneralOptimizationProblem}).

\begin{lemma}
	For any values of ${ \bar{A} }$ and $\bar{P}_\text{\upshape harv}^\text{\upshape req} \in \Big[0, \bar{P}_\text{\upshape max} \Big]$, the maximum mutual information $I^*_{\mathcal{F}_{ \bar{A} }}$ as a solution of (\ref{Eqn:GeneralOptimizationProblem}) is lower-bounded by 
	\begin{equation}
		I^*_{\bar{\mathcal{F}}_{ \bar{A} }} \triangleq \max_{f_s \in \bar{\mathcal{F}}_{ \bar{A} }} \, I(f_s) \geq \max_{f_s \in \bar{\mathcal{F}}_{ \bar{A} }} \; J(f_s) \triangleq J^*_{\bar{\mathcal{F}}_{ \bar{A} }},
		\label{Eqn:MI_LowerBound}
	\end{equation}
\noindent where $\bar{\mathcal{F}}_{ \bar{A} } = \{f_s \, \vert \, f_s \in \mathcal{F}_{ \bar{A} }, \bar{P}_\text{\upshape harv}(f_s) \geq \bar{P}^\text{\upshape req}_\text{\upshape harv} \}$ is the feasible set of (\ref{Eqn:GeneralOptimizationProblem}), $J(f_s) = \frac{1}{2} \ln\big(1 + \frac{e^{2 h_x(f_s)}}{2\pi e \sigma^2 }\big)$ is the achievable mutual information as function of pdf $f_s(\cdot)$, and $h_x(f_s)$ is the differential entropy of random variable $x = \sqrt{\psi( |h s|^2) }$ for a given $f_s(\cdot)$.
\label{Lemma:EPI}
\end{lemma}
\begin{proof}
\ifdefined\draftversion
	The corresponding proof is similar to that of \cite[Appendix A]{Lapidoth2009} and is omitted here due to the space constraints.
\else
	The proof follows along the lines of \cite[Appendix A]{Lapidoth2009}.
	Specifically, utilizing the entropy power inequality, we express the maximum mutual information in  (\ref{Eqn:GeneralOptimizationProblem}) as follows:
	\begin{align}
		I^*_{\bar{\mathcal{F}}_{ \bar{A} }} &\triangleq \max_{f_s \in \bar{\mathcal{F}}_{ \bar{A} }} \; I(f_s) = \max_{f_s \in \bar{\mathcal{F}}_{ \bar{A} }} \; h_y(f_s) - h_{n} \\
		&\geq \max_{f_s \in \bar{\mathcal{F}}_{ \bar{A} }} \; \frac{1}{2} \ln \big( e^{2 h_x(f_s)} + e^{2 h_n} \big) - h_{n} =  J^*_{\bar{\mathcal{F}}_{ \bar{A} }},
	\end{align}
	\noindent where $h_{y}(f_s)$ and $h_{n} = \frac{1}{2} \ln (2 \pi e \sigma^2)$ are the differential entropies of $y$ for a given pdf $f_s \in \bar{\mathcal{F}}_{ \bar{A} }$ and AWGN $n$, respectively.
	This concludes the proof.
\fi
\end{proof}

Lemma~\ref{Lemma:EPI} shows that the maximum mutual information $I^*_{\bar{\mathcal{F}}_{ \bar{A} }}$ as a solution of (\ref{Eqn:GeneralOptimizationProblem}) can be lower-bounded by the achievable mutual information $J^*_{\bar{\mathcal{F}}_{ \bar{A} }}$, which, in turn, is obtained as a solution of the optimization problem in (\ref{Eqn:MI_LowerBound}).
In the following, as a suboptimal solution of (\ref{Eqn:GeneralOptimizationProblem}), we determine $J^*_{\bar{\mathcal{F}}_{ \bar{A} }}$ for given ${ \bar{A} }$ and $\bar{P}_\text{\upshape harv}^\text{\upshape req} \in [0, \bar{P}_\text{\upshape max} ]$.
First, in the following proposition, we show that for small required average harvested powers $\bar{P}^\text{req}_\text{harv}$, constraint (\ref{Eqn:GeneralOptConstr1}) in the definition of $\bar{\mathcal{F}}_{ \bar{A} }$ can be relaxed and the achievable information rate $J^*_{\bar{\mathcal{F}}_{ \bar{A} }}$ can be obtained in closed form.

\begin{proposition}
	For a given $\bar{A} = \min \{A, \frac{ \sqrt{\rho_\text{\upshape max}} }{|h|}\}$ and required average harvested powers satisfying the constraint $\bar{P}^\text{\upshape req}_\text{\upshape harv} \leq \frac{1}{3} \bar{P}_\text{\upshape max},$ the achievable information rate is given by $J^*_{\bar{\mathcal{F}}_{ \bar{A} }} =  \frac{1}{2} \ln\big(1 + \frac{ \bar{P}_\text{\upshape max}}{2\pi e \sigma^2 }\big)$ and the corresponding pdf of $x$ is $f^*_x(x) = \frac{1}{\sqrt{\bar{P}_\text{\upshape max}}}$ with $\mathcal{D} \{f^*_x\} = [0, \sqrt{\bar{P}_\text{\upshape max}}]$.
	\label{Prop:UniformDistribution}
\end{proposition}
\begin{proof}
Please refer to Appendix~\ref{Appendix:UniformDistribution}.
\end{proof}

Proposition~\ref{Prop:UniformDistribution} demonstrates that if the required average harvested power $\bar{P}_\text{harv}^\text{req}$ is low, there is no tradeoff between the achievable information rate and the average harvested power and the corresponding $J^*_{\bar{\mathcal{F}}_{ \bar{A} }}$ can be computed in closed form.
In the next proposition, we consider the case $\bar{P}_\text{harv}^\text{req} \in [\frac{1}{3} \bar{P}_\text{\upshape max}, \bar{P}_\text{\upshape max}]$ and characterize the corresponding achievable information rate $J^*_{\bar{\mathcal{F}}_{ \bar{A} }}$.

\begin{proposition}
	For a given $\bar{A} = \min \{A, \frac{ \sqrt{\rho_\text{\upshape max}} }{|h|}\}$ and a required average harvested power $\bar{P}_\text{\upshape harv}^\text{\upshape req} \in [\frac{1}{3} \bar{P}_\text{\upshape max}, \bar{P}_\text{\upshape max}]$, the achievable information rate $J^*_{\bar{\mathcal{F}}_{ \bar{A} }}$ as a solution of the optimization problem in (\ref{Eqn:MI_LowerBound}) is given by \vspace*{-5pt}
	\begin{equation}
		J^*_{\bar{\mathcal{F}}_{ \bar{A} }} = \frac{1}{2} \ln\big(1 + \frac{e^{2 \mu_0 - 2\mu_2 \bar{P}_\text{\upshape harv}^\text{\upshape req} }}{2\pi e \sigma^2 }\big).
		\label{Eqn:OptMI}
	\end{equation}
\noindent  Here, $\mu_0 = \mu_2 \bar{P}_\text{\upshape max} + \ln(\frac{ \sqrt{\bar{P}_\text{\upshape max}} }{1 + 2\mu_2 \bar{P}_\text{\upshape harv}^\text{\upshape req} })$ and $\mu_2 \in \mathbb{R}_{+}$ is obtained as a solution of the following equation:
\begin{align}
	\ln(1 + 2 \mu_2 \bar{P}_\text{\upshape harv}^\text{\upshape req}) + \ln\big(Ei(\sqrt{\mu_2 \bar{P}_\text{\upshape max}} )\big) &= \nonumber \\
	\frac{1}{2} \ln\left(\frac{4 \bar{P}_\text{\upshape max} \mu_2 }{\pi} \right) &+ \mu_2 \bar{P}_\text{\upshape max},
	\label{Eqn:OptimalCoef}
\end{align}
\noindent where $Ei(\cdot)$ denotes the imaginary error function. Furthermore, the corresponding pdf of $x$ is $f^*_x(x) = \exp(- {\mu}_0 + {\mu}_2 x^2), x\in [0, \sqrt{\bar{P}_\text{\upshape max}} ]$,
\label{Prop:OptimalSolution}
\end{proposition}
\begin{proof}
Please refer to Appendix~\ref{Appendix:OptimalSolution}.
\end{proof}

Proposition~\ref{Prop:OptimalSolution} shows that if $\bar{P}_\text{\upshape harv}^\text{\upshape req} \in [\frac{1}{3} \bar{P}_\text{\upshape max}, \bar{P}_\text{\upshape max}]$, there is a tradeoff between the achievable information rate and the average harvested power and, for given ${ \bar{A} }$ and $\bar{P}_\text{\upshape harv}^\text{\upshape req}$, the achievable information rate $J^*_{\bar{\mathcal{F}}_{ \bar{A} }}$ is given by (\ref{Eqn:OptMI}), where $\mu_2$ is a solution of (\ref{Eqn:OptimalCoef}).
We note that although determining a closed-form solution of (\ref{Eqn:OptimalCoef}) is not feasible in general, $\mu_2$ can always be obtained numerically\footnotemark\hspace*{0pt} with any desired precision error via a one-dimensional grid search \cite{Coope2001}.
\footnotetext{In our exhaustive numerical simulations, for any ${ \bar{A} }$ and $\bar{P}_\text{\upshape harv}^\text{\upshape req}$, we could always find a unique solution of (\ref{Eqn:OptimalCoef}).}

%% file: SystemSetup.tex
\begin{table}[!t]
	\centering
	\caption{Tuned parameters of the EH model in (\ref{Eqn:EHmodel}).}
	\begin{tabular}{|m{0.2\textwidth} | m{0.2\textwidth}|}
		\hline 
		\multicolumn{2}{|c|}{$N = 2, \rho_0 = \SI{0}{\milli\watt}, \rho_1 = \SI{1.8}{\milli\watt}, \rho_\text{max} = \SI{2.4}{\milli\watt}$}\\
		\hline
		\multicolumn{1}{|c|}{$\varphi_1(\cdot)$} & 
		\multicolumn{1}{|c|}{$\varphi_2(\cdot)$} \\
		\hline
		\makecell[l]{$B = 7.16 \cdot 10^{-5}, \alpha = 1.432,$\\ $\beta = 0.778, \theta = 2174.86$ } & \makecell[l]{$B = 2.5 \cdot 10^{-5}, \alpha = 1.841,$\\ $\beta = 0.445, \theta = 956.75$ }\\
		\hline
	\end{tabular}
	\label{Table_EhModel}
	\vspace*{-10pt}
\end{table}
In the following, we determine the tradeoff between the mutual information and the average harvested power via simulations.
We assume a line-of-sight between the TX and RX and model the channel as $h = \tilde{h} \hat{h}$, where $\tilde{h} = \frac{c_{l}}{4 \pi d f_\text{c}} \sqrt{G_\text{T} G_\text{R} }$ and $\hat{h}$ are the large- and small-scale channel coefficients, respectively.
Here, $c_{l}$ and $f_\text{c} = \SI{300}{\giga\hertz}$ are the speed of light and carrier frequency, respectively.
Furthermore, to efficiently charge the microscopic THz IoT RX, we set the TX and RX antenna gains and distance between the TX and RX to $G_\text{T} = \SI{30}{\dBi}, G_\text{R} = \SI{10}{\dBi},$ and $d = \SI{0.3}{\meter}$, respectively.
We model the small-scale fading coefficient $\hat{h}$ as a Rician distributed random variable with Rician factor $1$ \cite{Morsi2019, Clerckx2019}.
The noise variance at the output of the RX is set to $\sigma^2 = \SI{-50}{\dBm}$.
We average all simulation results over $1000$ channel realizations.

First, we tune the parameters of the proposed general EH model $\psi(\rho)$ to match circuit simulation results.
To this end, we model the RTD-based EH circuit in Fig.~\ref{Fig:SystemModel} with the circuit simulation tool Keysight ADS \cite{ADS2017}.
In particular, we utilize the RTD design developed in \cite{Clochiatti2020}, which was matched to the measurement data reported in \cite{Clochiatti2020}, see Fig.~\ref{Fig:IV_Curve}.
The tuned EH model $\psi(\cdot)$ is shown in Fig.~\ref{Fig:MatchedEhModel} and the corresponding model parameters are summarized in Table~\ref{Table_EhModel}.

We observe in Fig.~\ref{Fig:MatchedEhModel} that $N=2$ functions $\varphi_n(\cdot)$ are sufficient to model $\psi(\cdot)$.
Furthermore, as expected, the instantaneous harvested power $P_\text{h}$ grows with the input power for $\rho \in [0, \rho_1)$.
Then, the instantaneous harvested power decreases when the input power $\rho \in [\rho_1, \rho_\text{max}]$ grows until the maximum power level $\rho=\rho_\text{max}$ is reached.
This behavior substantially differs from Schottky diode-based EH circuits, where the instantaneous harvested power is characterized by a monotonic non-decreasing function of the input power \cite{Kim2022, Clerckx2019, Boshkovska2015, Morsi2019}.
Finally, for high input powers exceeding $\rho = \rho_\text{max}$, the RTD is driven into breakdown, which has to be avoided since, when operating in the breakdown regime, the diode may be destroyed by a large reverse-bias current \cite{Tietze2012}, as shown in Fig.~\ref{Fig:IV_Curve}.

In Fig.~\ref{Fig:Results_Tradeoffs}, we plot the mutual information and average harvested power for different values of the maximum amplitude of the transmitted signal ${ \bar{A} }$.
To this end, for different values of required harvested power $\bar{P}_\text{harv}^\text{req} \in [0, \frac{1}{3} \bar{P}_\text{\upshape max}]$ and $\bar{P}_\text{harv}^\text{req} \in [\frac{1}{3} \bar{P}_\text{\upshape max}, \bar{P}_\text{\upshape max}]$, we first obtain the maximum achievable rates $J^*_{\bar{\mathcal{F}}_{ \bar{A} }}$ in Propositions~\ref{Prop:UniformDistribution} and \ref{Prop:OptimalSolution}, respectively.
Next, for the pdfs $f_s^*$ that yield $J^*_{\bar{\mathcal{F}}_{ \bar{A} }}$, we also plot in Fig.~\ref{Fig:Results_Tradeoffs} the mutual information $I(f_s^*)$.
Also, as a baseline scheme, transmit symbols follow a Gaussian distribution truncated to the interval $[0, { \bar{A} }]$ and with mean $\frac{{ \bar{A} }}{2}$.
In particular, for a given value of ${ \bar{A} }$, the corresponding values of mutual information $I(f_s)$ and average harvested power $\bar{P}_\text{harv}(f_s)$ for the baseline scheme in Fig.~\ref{Fig:Results_Tradeoffs} are obtained by adjusting the variance $\sigma_\text{s} \in [0, +\infty)$ of the truncated Gaussian distribution~$f_s$.

\begin{figure}[!t]
	\centering
	\includegraphics[draft = false, width=0.42\textwidth]{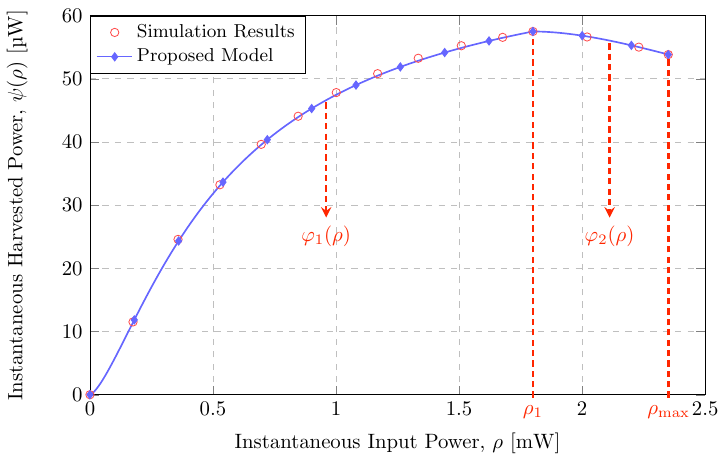}
	\vspace*{-7pt}
	\caption{Proposed EH model tuned to match the circuit simulation results.}
	\label{Fig:MatchedEhModel}
	\vspace*{-10pt}
\end{figure}

We observe in Fig.~\ref{Fig:Results_Tradeoffs} that for any given required average harvested power $\bar{P}_\text{harv}^\text{req}$, the proposed SWIPT system is able to achieve significantly higher information rates than the baseline scheme.
Thus, we conclude that Gaussian distributed signals are highly suboptimal, and an optimal design of the transmit symbol waveform is needed for efficient THz SWIPT.
Furthermore, we also observe in Fig.~\ref{Fig:Results_Tradeoffs} that for all ${ \bar{A} }$ and $\bar{P}_\text{harv}^\text{req}$, the gap between the mutual information $I(f_s^*)$ and the achievable rate $J^*_{\bar{\mathcal{F}}_{ \bar{A} }}$ is small.
For all considered values of ${ \bar{A} }$, both the maximum achievable rate $J^*_{\bar{\mathcal{F}}_{ \bar{A} }}$ and the mutual information $I(\cdot)$ decrease as the required average harvested power increases.
Thus, for any peak amplitude ${ \bar{A} }$, there is a tradeoff between the achievable information rate $J^*_{\bar{\mathcal{F}}_{ \bar{A} }}$ and the harvested power $\bar{P}_\text{harv}(\cdot)$ that is characterized by the rate-power region in Fig.~\ref{Fig:Results_Tradeoffs}.
Moreover, we note that for low maximum transmit signal amplitudes, which satisfy ${ \bar{A} } < \frac{\sqrt{\rho_1}}{|h|}$, both the mutual information and the harvested power grow with ${ \bar{A} }$.
However, if the maximum received power at the RX satisfies $ { \bar{A} } \geq \frac{\sqrt{\rho_1}}{|h|}$, i.e., for ${ \bar{A} } \in \{\SI{0.75}{\volt}, \SI{1}{\volt}\}$ in Fig.~\ref{Fig:Results_Tradeoffs}, the boundary of the rate-power regions does not change, and thus, the tradeoff between the mutual information and average harvested power of the THz SWIPT system is determined by the maximum instantaneous harvested power $\psi(\rho_1)$ in Fig.~\ref{Fig:MatchedEhModel} and not by the value of ${ \bar{A} }$.
\vspace*{-5pt}

%% file: Conclusions.tex
In this work, we studied SWIPT for micro-scale 6G THz IoT systems and proposed RTDs for EH and unipolar ASK for information and power transfer.
Furthermore, we developed a general non-linear piecewise EH model, whose parameters were tuned to fit circuit simulation results.
Based on this new model, we formulated an optimization problem for the maximization of the mutual information between the TX and RX signals subject to constraints imposed on the peak amplitude of the transmitted signal and the average harvested power at the IoT RX.
We derived a feasibility condition for this optimization problem, and for high and low required harvested powers, we derived the achievable information rate numerically and in closed form, respectively.
Our simulation results revealed a tradeoff between the achievable mutual information and the average harvested power.
Finally, we observed that the proposed TX signal design significantly outperforms a baseline scheme where truncated non-negative centered Gaussian distributed transmit symbols were employed.

\begin{figure}[!t]
	\centering
	\includegraphics[draft = false, width=0.44\textwidth]{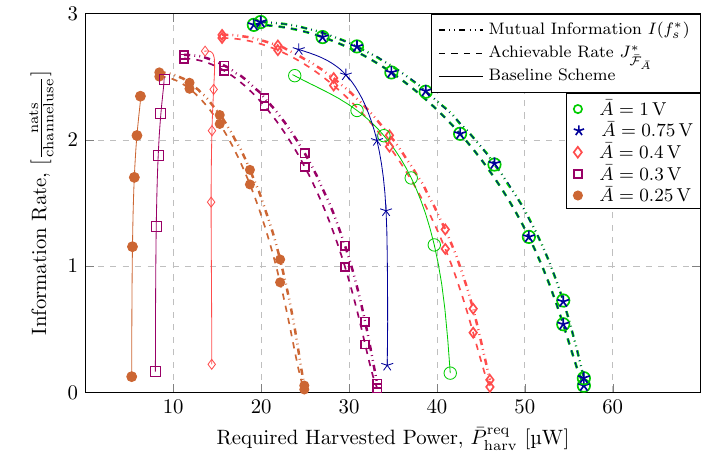}
	\vspace*{-7pt}
	\caption{Achievable rate-power regions for different values of $\bar{A}$.}
	\label{Fig:Results_Tradeoffs}
	\vspace*{-10pt}
\end{figure}

%% file: ProofFeasibility.tex
We note that for a given maximum transmit signal amplitude $\bar{A}$, the instantaneous harvested power can not exceed $\bar{P}_\text{\upshape max}$.
Then, the average power harvested at the RTD-based RX is upper-bounded by $\max_{f_s \in {\mathcal{F}}_{\bar{A}}} \bar{P}_\text{harv}(f_s) = \bar{P}_\text{\upshape max}$.
Thus, for any $\bar{P}^\text{req}_\text{harv} > \bar{P}_\text{\upshape max}$, a solution of (\ref{Eqn:GeneralOptimizationProblem}) does not exist.
On the other hand, for any $\bar{P}^\text{req}_\text{harv} \in[0,  \bar{P}_\text{\upshape max}]$, there exists at least one pdf $f^1_s = \delta(s-s_0) \in \mathcal{F}_{\bar{A}}$, where $s_0 \leq {\bar{A}}$ is chosen such that $\psi(|hs_0|^2) = \bar{P}_\text{harv}^\text{req}$, which satisfies constraint (\ref{Eqn:GeneralOptConstr1}) with equality.
This concludes the proof.

%% file: ProofUniformDistribution.tex
First, we note that if the average power constraint (\ref{Eqn:GeneralOptConstr1}) in the definition of $\bar{\mathcal{F}}_{\bar{A}}$ is not present, the differential entropy $h_x(f_s)$, and hence, function $J(f_s)$ with $f_s \in \bar{\mathcal{F}}_{\bar{A}}$ are maximized if the pdf of $x$ is uniform and given by $f_x^*$ \cite{Lapidoth2009}.
Furthermore, in this case, the maximum achievable information rate and the average harvested power can be expressed as $J(f_s) = J^*_{\bar{\mathcal{F}}_{\bar{A}}}$ and $\bar{P}^\text{\upshape req}_\text{\upshape harv} = \mathbb{E}_x\{x^2\} = \frac{1}{3} \bar{P}_\text{\upshape max}$, respectively.
Thus, for $\bar{P}^\text{\upshape req}_\text{\upshape harv} \leq \frac{1}{3} \bar{P}_\text{\upshape max}$, constraint (\ref{Eqn:GeneralOptConstr1}) can be relaxed for the maximization of $J(\cdot)$ and the optimal distribution of $s$ that yields $J^*_{\bar{\mathcal{F}}_{\bar{A}}}$ can always be obtained.
This concludes the proof.

%% file: ProofOptimalSolution.tex
First, since $x$ is a deterministic function of input signal $s$, we obtain $J^*_{\bar{\mathcal{F}}_{\bar{A}}}$ for given ${\bar{A}}$ and $\bar{P}^\text{\upshape req}_\text{\upshape harv} \geq \frac{1}{3} \bar{P}_\text{\upshape max}$ as follows:\vspace*{-5pt}
\begin{equation}
	J^*_{\bar{\mathcal{F}}_{\bar{A}}} = \max_{f_x \in \hat{\mathcal{F}}} \; J_x(f_x),
	\label{Eqn:Prop3Eqn1} \vspace*{-5pt}
\end{equation}
\noindent where $\hat{\mathcal{F}} = \{f_x \, \vert \, \mathcal{D}\{f_x\} \in [0,\sqrt{\bar{P}_\text{\upshape max}}], \int_{x}f_x(x)\, \text{d}x = 1, \mathbb{E}_x\{x^2\} \geq \bar{P}^\text{\upshape req}_\text{\upshape harv} \}$ and \vspace*{-5pt}
\begin{equation}
	J_x(f_x) = \frac{1}{2} \ln\big(1 + \frac{e^{2 \tilde{h}_x(f_x)}}{2\pi e \sigma^2 }\big).
	\label{Eqn:Prop3Eqn2}\vspace*{-5pt}
\end{equation}
\noindent\hspace*{0pt}Here, $\tilde{h}_x(f_x) = -\int_x f(x) \ln\big(f(x)\big) \text{d}x$ is the differential entropy of $x$ expressed as a function of pdf $f_x$.
Furthermore, since function $J_x(\cdot)$ is monotonically increasing in $\tilde{h}_x(\cdot)$, for $\bar{P}_\text{\upshape harv}^\text{\upshape req} \in [\frac{1}{3} \bar{P}_\text{\upshape max}, \bar{P}_\text{\upshape max}]$, the pdf $f^*_x \in \hat{\mathcal{F}}$ solving (\ref{Eqn:Prop3Eqn1}) is the maximum entropy distribution, i.e., $f^*_x$ yields the maximum value of $\tilde{h}_x(\cdot)$ among all pdfs satisfying $\mathcal{D}(f_x) \in [0, \sqrt{\bar{P}_\text{\upshape max}}]$ and $\mathbb{E}_x\{x^2\} \geq \bar{P}^\text{\upshape req}_\text{\upshape harv}$ \cite{Lapidoth2009}.
Exploiting the Karush–Kuhn–Tucker (KKT) conditions, it can be shown that the optimal pdf as solution of (\ref{Eqn:Prop3Eqn1}) is given by $f^*_x(x)$ in Proposition~\ref{Prop:OptimalSolution}.
Moreover, the corresponding differential entropy of $x$ is given by\vspace*{-5pt}
\begin{equation}
	h_x(f_x^*) = -\int_{x} f^*_x(x) \ln(f^*_x(x)) \text{d} x = {\mu}_0 - {\mu}_2 \bar{P}^\text{\upshape req}_\text{\upshape harv}.
	\label{Eqn:Prop3Eqn3}\vspace*{-5pt}
\end{equation}
Finally, substituting (\ref{Eqn:Prop3Eqn3}) into (\ref{Eqn:Prop3Eqn2}), we obtain (\ref{Eqn:OptMI}).
This concludes the proof.